\documentclass[11pt]{article}
\usepackage{amsmath, amsthm, amssymb}
\usepackage[margin=1in]{geometry}
\usepackage{graphicx}
\newcommand{\bra}[1]{\langle #1|}
\newcommand{\ket}[1]{|#1\rangle}

\newtheorem{thm}{Theorem}[section]

\theoremstyle{remark}

\theoremstyle{definition}

\DeclareMathOperator{\polylog}{polylog}

\begin{document}

\author{Michael Ben Or\thanks{E-mail: benor@cs.huji.ac.il} \\ The Hebrew University\\Jerusalem, Israel \and Daniel Gottesman\thanks{E-mail: dgottesman@perimeterinstitute.ca} \\ Perimeter Institute \\ Waterloo, Canada \and Avinatan Hassidim\thanks{E-mail: avinatanh@gmail.com} \\ Bar-Ilan University \\Ramat Gan, Israel}

  \title{Quantum Refrigerator}
  \maketitle

\begin{abstract}
We consider fault-tolerant quantum computation in the context where there are no fresh ancilla qubits available during the computation, and where the noise is due to a general quantum channel.  We show that there are three classes of noisy channels: In the first, typified by the depolarizing channel, computation is only possible for a logarithmic time.  In the second class, of which the dephasing channel is an example, computation is possible for polynomial time.  The amplitude damping channel is an example of the third class, and for this class of channels, it is possible to compute for an exponential time in the number of qubits available.
\end{abstract}

\section{Introduction}
\label{sec:intro}

The threshold theorem~\cite{AB97,KLZ98,Kit97,AGP05} is the central result of the theory of fault-tolerant quantum computation.  It states that, provided the error rate per gate or time step is below some constant threshold value, then arbitrarily long quantum computations are possible with only polylogarithmic overhead.  The threshold theorem tells us that large quantum computers are possible in principle, provided experimentalists can achieve an error rate below the threshold value.

Naturally, then, it is of great interest to determine the precise physical assumptions which are needed for a threshold error rate to exist.  Some assumptions made in the simplest versions of the threshold theorem can certainly be relaxed.  For instance, a threshold exists when gates are restricted to nearest-neighbor interactions~\cite{AB08,Got00}, or when the noise is due to very general weak interactions with a non-Markovian environment~\cite{AGP05,TB05}.  Other assumptions are truly needed.  For instance, it must be possible to perform gates in parallel, and the noise process must not cause many qubits to fail simultaneously in a correlated way.

One assumption that has been traditionally classified as required is the need for fresh ancilla qubits in the course of the computation.  In the course of quantum error correction, ancilla qubits are introduced and used to record the error syndrome.  This can be viewed as a refrigeration process, where entropy which has been introduced into the data qubits by the noise gets pumped out into the ancilla qubits, cooling down the data qubits.  In order for this to work, the ancilla qubits used must be cold themselves, or they cannot absorb the extra entropy from the data.  Since ancilla qubits created at the beginning of the computation are themselves subject to the noise process, we can expect them to heat up over time, eventually making them worthless for refrigeration.  Based on this intuition, we expect to need a continual stream of new, freshly cooled ancilla qubits to keep the error correction running.

At a rigorous mathematical level, this intuition is supported by the result of \cite{ABIN96}, which showed that for noise in the form of a depolarizing channel, it is impossible to compute for longer than $O(\log n)$ time.  We study more general channels and reach a different conclusion.  In particular, we prove the following main theorem.  The precise statements for what the upper and lower bounds mean in each case are deferred to the appropriate sections.

\begin{thm}[Main Theorem]
Let $C$ be any non-unitary channel close to the identity, and consider quantum computations which suffer from noise $C$ on each qubit at each time step.
Up to unitary equivalence, the limit of repeatedly applying $C$ can be either a point or a diameter in the Bloch sphere.\footnote{If $C$ is a dephasing channel followed by a unitary rotation, the limit point of $C$ applied repeatedly may be the center of the Bloch sphere, but we still consider it to be part of the 2nd class.  Therefore, we choose the largest set of fixed points of a channel unitarily equivalent to $C$.  See section~\ref{sec:classification} for more details.}
\begin{enumerate}

\item If the limit is the center of the Bloch sphere, it is possible to
compute for $\tilde{O}(\log(n))$ time steps. This is tight up to $\log \log n$ factors.  An example of this class is the depolarizing channel.

\item If the limit is a diameter, than it is possible to compute on $O(n^a)$ qubits for $O(n^b)$ time steps provided $a + b<1$,
and impossible even to store a single (unknown) qubit for more than $O(n^3)$ time steps.  An example of this class is the dephasing channel.

\item If the limit is a point which is not the center of the Bloch sphere, it
is possible to compute for an exponential number of time steps.  An example of this class is the amplitude damping channel.
\end{enumerate}
\end{thm}

For the third case, we expect that it is not possible to compute for longer than exponential time, but we have not managed to prove this.

The three classes of channels can also be characterized in terms of how they treat the entropy of a state.  The depolarizing class causes entropy to increase for all states until it reaches the maximum for the completely mixed state.  For the dephasing class of channels, entropy is merely non-decreasing: For some states it increases, and for others it stays the same.  For the amplitude damping class of channels, entropy can decrease under the channel.

The result with the greatest practical significance is 3, as it enables quantum computation without fresh qubits.  This doesn't really violate the intuition discussed above, it is just that for this class of channels, we can use the noise itself to cool the ancilla qubits.
The rest of the paper is devoted to proving the main theorem.
Here we give some intuition to the proofs of 1-3.

\begin{itemize}
\item {\bf Limit($C$) is the center of the Bloch sphere: }
The proof is just a slight generalization of \cite{ABIN96}, and
uses the same techniques.

\item {\bf Limit($C$) is a diameter: }
The noise is equivalent
to the phase damping channel. Computing for $O(n^b)$ time steps
can be done by running the computation on $\tilde{O}(n^a)$ qubits,
leaving $n - \tilde{O}(n^a)$ qubits in the state $\ket{0}$
(or some other pure state which is unaffected by the noise).
At each stage of the computation it is possible
to take fresh qubits which were unaffected by the noise.
Since at most $n^a$ qubits are needed every step,
the computation can go for time $n^b$ (theorem \ref{thm:dephasinglower}).

Proving the other direction is a little tricky, as it
is not clear how fast the entropy rises. We will show that it is impossible to store half of an EPR pair for more than $O(n^3)$ time.  To do this, we take advantage of the fact that the entropy cannot decrease.  Therefore, if the state is stored for a long time, there must be some time step during which the entropy rises very little.  For the dephasing channel, this can only happen if the qubits are essentially unentangled and diagonal in the standard basis.  In particular, further dephasing will have little effect on them.  But a completely dephased state has no entanglement with the outside world at all, and therefore we have failed to maintain the EPR pair.

\item {\bf Limit($C$) is a point which is not the center of the sphere: }
In this model we use $\tilde{O}(n)$ qubits instead of $n$ qubits, and divide the system into
three parts. 

The first part is just a standard compilation of the original circuit to a fault-tolerant computation, just as if fresh ancillas were available. Denote by $n'$ the number of qubits which the fault tolerant computation uses. We have that $n' = \tilde{O}(n)$, and that at each step the fault tolerant circuit asks for at most $n'$ fresh ancilla qubits.

The second part of the computation is a storage house. This part
holds $n' T$ chunks of data. The size of each chunk is $R$, and
at each stage of the computation $n'$ chunks are passed to the refrigerator.

The last part is the refrigerator, which is actually built of $O(n')$
refrigeration units which work in parallel. Each such unit
takes a typical sequence of $R$ qubits (where each qubit in the sequence
is very similar to the limit of $C$),
and condenses the entropy to $R - 1$ qubits and to a
fresh qubit. Then the refrigerator passes $n'$ qubits to
the computation, which uses as many as it needs.
The dirty qubits left over from the refrigerator, the qubits used up by the computation, and any unused
fresh qubits are recycled by being sent back to the storage house.

The constant $R$ is chosen so that it will be possible to
condense the entropy. $T$ is chosen
such that after $T$ applications of $C$ the qubit will
be close enough to the limit of $C$. We take the noise to
be weak enough so that not too many mistakes
occur during the process in which the refrigerator condenses the entropy.

\end{itemize}

\subsection{Previous results}

We only state a few fundamental results in this field
which are important for our constructions. The first and foremost
important results are the threshold proofs. \cite{AB97} and \cite{AGP05} prove that quantum
computation is possible under any noisy channel if the channel is
weak enough (but still a constant), and if fresh qubits are
introduced. \cite{KLZ98} and \cite{Kit97} also proved the threshold theorem,
but only for stochastic noise.
\begin{thm}[Threshold Theorem]
There exists a constant \emph{threshold} value $p_T > 0$ with the following property: Suppose a computation undergoes noise $C$ at
each location (each gate and each time step), with $\|C - I \|_\Diamond < p_T$. Then given any
quantum circuit with depth $t$ acting on $n$ qubits, there exists a fault-tolerant simulation of the circuit whose output has
statistical distance at most $\epsilon$ away from that of the ideal circuit.  The fault-tolerant simulation uses
$O(n \polylog (nt/\epsilon))$ qubits and has depth $O(t \polylog (nt/\epsilon))$.
\end{thm}

An assumption of the threshold theorem is that fresh ancilla qubits can be supplied at each step.
If those ancilla qubits are not available, the threshold theorem no longer directly applies.
$\| \cdot \|_\Diamond$ is the diamond norm~\cite{KSV02}, equivalent in this case to the completely bounded operator
norm.

Our paper can be viewed as a generalization of \cite{ABIN96}. \cite{ABIN96} showed that for the
depolarizing channel, computation is impossible for more than
$O(\log(n))$ steps. Razborov~\cite{Raz06} used a similar argument
to show that a very noisy channel is useless. These two
results could create the false impression that quantum computation
is impossible without the introduction of fresh qubits. To the best
of our knowledge, phase damping, amplitude damping, or more complex
kinds of noise did not receive a full treatment.

\section{The Model}

We will work in a model in which gates are perfect, but after each time step, each qubit undergoes a single-qubit channel $C$.  Thus, we alternate layers consisting of gates with layers consisting of noise.  Every qubit undergoes the same noise at each time step, and there is no memory or correlation between different qubits or different time steps.  One could consider more complicated models in which qubits undergo different kinds of noise depending on the gates performed on them (i.e., allowing gate errors), or in which different qubits have different error models, but we believe this would simply add complication without fundamentally changing the results.

The lower bounds (showing that computation is possible for a certain amount of time) are proven by allowing arbitrary single-qubit and two-qubit gates in parallel at each time step, as is usual for fault-tolerant constructions.  The upper bounds (showing that computation is impossible for more than a certain amount of time) are more generous, and allow arbitrary many-qubit unitaries between the applications of noise at each time step.

Because the only control available to us is via unitary gates, it is not possible to measure error syndromes and perform corrections based on classical processing of the outcome.  However, \cite{AB08} showed that a threshold still exists when all control is via unitary gates and measurement in the middle of the computation is not possible.  Therefore, we will use this version of the threshold theorem throughout the paper.

\section{Classification of Qubit Channels}
\label{sec:classification}

Using the Bloch sphere representation of a qubit, any density matrix $\rho$ can be written as
\begin{equation}
\rho = \frac{1}{2} (I + \mathbf{w} \cdot \mathbf{\sigma}),
\end{equation}
where $\mathbf{w}$ is a real vector with norm $\leq 1$ (i.e., inside the Bloch sphere) and $\mathbf{\sigma}$ is the vector of Pauli matrices $(X, Y, Z)$.  In this representation, the action of any quantum channel $C$ can be written~\cite{KR01} as $C(\rho) = U C'(V \rho V^\dagger) U^\dagger$, with
\begin{equation}
C' \left(\frac{1}{2} [I + \mathbf{w} \cdot \mathbf{\sigma}] \right) = \frac{1}{2} [I + (\mathbf{t} + T \mathbf{w}) \cdot \sigma],
\end{equation}
such that $T$ is diagonal.  $\mathbf{t} = 0$ iff the channel $C$ is unital (meaning $C(I) = I$).  That is, up to unitaries $U$ and $V$, we can regard the channel as shifting the Bloch vector $\mathbf{w}$ and rescaling it in the $X$, $Y$, and $Z$ coordinate axes.  In order for the channel to be positive, the output vector $\mathbf{t} + T \mathbf{w}$ must also have norm at most $1$.

In the future, we will ignore the unitaries $U$ and $V$.  This is because they are single-qubit unitaries and can be absorbed into the gate layers before and after the noise.  Furthermore, if $\mathbf{t}=0$ and $T=I$, $C$ is actually unitary and there is no noise.

There are unital maps (the $\mathbf{t}=0$ case) for which all three eigenvalues of $T$ are less than $1$.  In this case, all vectors $\mathbf{w}$ shrink towards the origin, and the only fixed point of the map is the center of the Bloch sphere, the maximally mixed state.  This gives us the first class of maps, including the \emph{depolarizing channel}
\begin{equation}
C(\rho) = (1-p) \rho + \frac{p}{3} (X \rho X + Y \rho Y + Z \rho Z).
\end{equation}
In fact, a map in this class is always a \emph{Pauli channel}, with probabilities $p_X$, $p_Y$, and $p_Z$ of errors $X$, $Y$, and $Z$:
\begin{equation}
C(\rho) = (1-p_X - p_Y - p_Z) \rho + p_X X \rho X + p_Y Y \rho Y + p_Z Z \rho Z.
\end{equation}
We can express $p_X$, $p_Y$, and $p_Z$ in terms of the eigenvalues $\lambda_X$, $\lambda_Y$, and $\lambda_Z$ of $T$, for instance
\begin{equation}
p_X = \frac{1}{4} (1+\lambda_X - \lambda_Y - \lambda_Z).
\end{equation}
By~\cite{KR01}, $|\lambda_Y + \lambda_Z| \leq |1+\lambda_X|$ for a completely positive map, so $p_X \geq 0$.  Therefore any unital qubit channel corresponds to a sensible Pauli channel, up to unitary rotations.

Next, we can consider unital channels for which $T$ has one or two eigenvalues equal to $1$.  We may assume without loss of generality (again, up to unitary rotations) that the $Z$ direction is one that is not contracted, so the $Z$ eigenvalue $\lambda_Z = 1$.  Then, by~\cite{KR01}, in order for $C$ to be completely positive, $|\lambda_X - \lambda_Y| \leq |1-\lambda_Z| = 0$.  Therefore, only one eigenvalue can be $1$ and the other two directions must contract by the same amount.  This can be achieved with the \emph{dephasing channel}
\begin{equation}
C(\rho) = (1-p) \rho + p Z \rho Z.
\end{equation}
These channels leave the $Z$ axis fixed and all other points contract towards it.  This is the second class of channels.

Finally, we have non-unital channels.  This is the third class of channels.  One such channel is the \emph{amplitude damping channel}
\begin{equation}
C(\rho) = \begin{pmatrix} 1 & 0 \\ 0 & \sqrt{1-p} \end{pmatrix} \rho \begin{pmatrix} 1 & 0 \\ 0 & \sqrt{1-p} \end{pmatrix} + \begin{pmatrix} 0 & \sqrt{p} \\ 0 & 0 \end{pmatrix} \rho \begin{pmatrix} 0 & 0 \\ \sqrt{p} & 0 \end{pmatrix}.
\end{equation}
A qubit non-unital channel always has a unique fixed point:  Since $\mathbf{t} \neq \mathbf{0}$ and the output vector $\mathbf{t} + T \mathbf{w}$ must have norm at most $1$ for all $\mathbf{w}$, it follows that $T$ has all eigenvalues strictly less than $1$, so the map is strictly contractive.  Therefore it has a unique fixed point, which we can identify as the vector
\begin{equation}
\left( \frac{t_X}{1-\lambda_X}, \frac{t_Y}{1-\lambda_Y}, \frac{t_Z}{1-\lambda_Z} \right).
\end{equation}

\section{Depolarizing Class}\label{sub-unital-point}

In this section we prove that $n$ qubits suffice for a circuit of
depth $\tilde{\Theta}(\log(n))$.

For the lower bound, we invoke a result from \cite{ABIN96}:

\noindent {\bf Theorem 4 from \cite{ABIN96}.} {\it If a boolean function $f$ can be computed by a quantum circuit of size $s$ and depth $d$,
then $f$ can be computed by a noisy quantum circuit of size $O(s \polylog(s)) \cdot 2^{O(d \polylog(d))}$ and depth
$O(d  \polylog(d))$.}

The theorem is stated only for depolarizing noise, but their proof works for any weak Pauli channel, which, as discussed in section~\ref{sec:classification}, covers all unital channels.  Basically, the argument is to continually purify ancillas by comparing them to each other with majority gates.  The number of reliable ancillas available drops exponentially with time, which produces the exponential factor in the size.  The purified ancillas are used to run a regular fault-tolerant protocol.  The consequence of this theorem is that using $n$ qubits, it is possible to run any computation of $\tilde{O}(\log n)$ depth.

To show that we can compute for a time at most $O(\log n)$, we can use essentially
the same proof as for Theorem 3 of \cite{ABIN96}.  Let $p = \min(p_X, p_Y, p_Z)$.
Let $D_p$ be the depolarizing channel with probability $p$.  $D_p$ is no more noisy than the actual
Pauli channel $C$, and by concavity of the entropy, the entropy increases at least as much under
$C$ as it does under $D_p$.  By lemma~8 of \cite{ABIN96}, the information $I(\rho) = n-S(\rho)$
decreases by at least a factor $1-p$ under $D_p$, so it also decreases by at least a factor $1-p$
under $C$.  The rest of the proof is the same as in \cite{ABIN96}.

\section{Dephasing channel}

In this section we consider a dephasing channel which acts on each qubit independently. We give polynomial upper and lower bounds for the ability to compute.

Consider a dephasing channel $\rho \rightarrow (1 - p) \rho + Z \rho Z^\dagger$.
\begin{thm}
\label{thm:dephasinglower}
  There exists  a constant $p_T$, such that
if $p \le p_T$, then one can reliably perform any computation which involves $O(n^a)$ qubits for $O(n^b)$ time, provided $a + b < 1$.
\end{thm}

\begin{proof}

The standard techniques of fault tolerance can be used to compile a circuit into a fault tolerant circuit, provided the noise is below some
threshold $p_T$ and that there is a supply of fresh qubits. If the original circuit acted on $m$ qubits and runs for a time $t$, the fault tolerant one acts on $m' = c_0 m \polylog (mt)$ qubits and has depth $t' = c_1 t \polylog(mt)$.

Suppose $m = O(n^a)$ and $t = O(n^b)$.  Given $n$ physical qubits, put aside $n-m'$ of them, and run the fault-tolerant simulation of the circuit
on the remaining $m'$.  The qubits that have been set aside begin in the state $\ket{0}$, so are unaffected by the dephasing
channel.  At each time step, the fault-tolerant circuit calls for up to $m'$ fresh ancillas.  We can supply them out of the qubits
set aside at the beginning.  After a qubit is used as an ancilla, it is not used again.  Therefore, we run out of fresh ancillas after
a time $(n-m')/m' = \Theta(n^{1-a}/ \polylog(mt))$.  This is sufficiently long to run the full computation if $t' \leq (n-m')/m'$, in other words,
if $O(n^b \polylog(mt)) \leq \Theta(n^{1-a}/\polylog(mt))$.  This will be true if $a + b < 1$.
\end{proof}

We also prove an impossibility result:

\begin{thm}
Begin with an EPR pair, with one qubit in a perfect reference system $R$ and the other qubit stored in a noisy quantum computer,
which undergoes a dephasing channel acting on each qubit with probability $p$ at each time step.
For any $p > 0$, after a time $T = \frac{(\ln 2) n^3}{8 p (1 - p)\epsilon^2}$, any quantum circuit applied to the noisy system will
still lead to a state which has distance (in the $2$-norm) at most $\epsilon$ from a separable state with the reference system.
\end{thm}

\begin{proof}

The proof is by contradiction. Suppose that there is a computation which could
protect half an EPR pair against dephasing noise, such that the original qubit could
be recovered (with high fidelity) after $T$ steps.  We track the evolution of the computation.

At any time step, let $q_1, \ldots, q_n$ denote the qubits before the noise
is applied, and let $q_1', \ldots, q_n'$ denote the qubits after the noise. $R$ is unaffected by the noise.  We
have that

\begin{equation}
S(q_1', \ldots, q_n',R) = \sum_i S(q_i'|q_{i+1}',q_{i+2}', \ldots, q_n',R) + S( R)
\geq \sum_i S(q_i'|q_{i+1}, \ldots, q_n,R) + S( R)
\end{equation}
because conditional entropy $S(A|B)$ cannot decrease under a CP map on B (which can only lose information about A).
Furthermore, $ S(q_i'|q_{i+1},\ldots ,q_n,R) \geq S(q_i|q_{i+1},\ldots,q_n,R)$ since a dephasing channel cannot decrease the entropy.
This means
\begin{equation}
 S(q_1', \ldots, q_n',R) \geq S(q_1, \ldots, q_n,R) + [S(q_1'|q_2, \ldots, q_n,R) - S(q_1|q_2, \ldots, q_n, R)].
\label{eqn:entropyincrease}
\end{equation}
Naturally, this holds for any order of the qubits.

Let $\delta = \frac{8 p (1 - p)\epsilon^2}{(\ln 2) n^2}$. At every step of the computation, if there is at least one qubit $q_i$ for which $S(q_i'| q_{-i})
> S(q_i|q_{-i}) + \delta$, then the entropy of the whole state
increases by at least $\delta$ by (\ref{eqn:entropyincrease}).
Here $S(q_i|q_{-i})$ means the conditional entropy of qubit $q_i$ conditioned on all other qubits in the system, including $R$.

However, since the entropy is bounded by $n$, and the computation stores the encoded half of
an EPR pair for $T = n / \delta$ time, there must be a time step $t$ in which the entropy increase is at most $\delta$. At this
time step, it must be that for every qubit $q_i$
\begin{equation}
S(q_i'|q_{-i}) - S(q_i|q_{-i}) \leq \delta
\end{equation}

We show that the decoding procedure will fail if applied after time $t$, by showing that at time
$t$ the state on the $n$ qubits is close to diagonal, and therefore the entanglement
to the reference system (the other half of the EPR pair encoded by the computation) is lost.

At time $t$ we have for every qubit $q_i$
\begin{equation}
S(q_i'|q_{-i}) - S(q_i|q_{-i}) = S( p \rho + (1-p) Z_i \rho Z_i) -
S(\rho),
\end{equation}
where $\rho$ is the joint state of all qubits in the computer and $R$.
At this point, concavity of the entropy tells us that the RHS must be non-negative, but we wish to set a
tighter bound in order to compare with the upper bound of $\delta$.

Pinsker's inequality (Theorem $8.6$ in \cite{Watrous})
gives a lower bound on the relative entropy:
\begin{equation}
S(\rho || \xi) \geq \frac{1}{2 \ln 2} \|\rho - \xi\|_2^2
\end{equation}
Concavity of the entropy can be proved from sub-additivity, and sub-additivity follows from the
positivity of the relative entropy, so Pinsker's inequality can give us a tighter version of concavity:
\begin{equation}
S[(1-p)\rho + p\xi] - [(1-p)S(\rho) + pS(\xi)] \geq \frac{2 p(1-p)}{\ln 2} \|\rho - \xi\|_2^2.
\end{equation}
Using this, we get that for any state $\xi$, if  $S[(1-p)\rho +
p \xi] - [(1- p)S(\rho) + p S(\xi) ] \le \delta$, then
\begin{equation}
\| \rho - \xi \|_2^2 \le \frac{(\ln 2) \delta}{2p(1-p)}.
\end{equation}
Plugging in $\xi = Z_i \rho Z_i$ gives
\begin{equation}
\| \frac{1}{2} (\rho + Z_i \rho Z_i) - \rho \|_2  = \frac{1}{2} \| Z_i \rho Z_i - \rho
\|_2 \le \frac{1}{2} \sqrt{(\ln 2) \delta / (2p(1-p))},
\end{equation}
so dephasing $\rho$ on qubit $i$ doesn't change it much. Since this holds for
every $i$, we can apply the triangle inequality and get
\begin{equation}
 \|\sigma  - \rho \|_2 \le  n/2 \sqrt{(\ln 2) \delta / (2p(1-p))},
\end{equation}
where $\sigma$ is the diagonal density matrix obtained by completely
dephasing $\rho$ on every qubit in the computer (but not $R$). Plugging in $\delta = \frac{8 p (1 - p)\epsilon^2}{(\ln 2) n^2}$ gives that
\begin{equation}
\|\sigma  - \rho \|_2 \le \epsilon.
\end{equation}
A perfect decoding map applied to the system at this point to try to extract the EPR pair cannot change the
entanglement with the reference system R.  Since the completely dephasing map is an entanglement-breaking
channel, $\sigma$ decodes to a separable state with R.  Therefore, $\rho$ decodes to a state with $2$-norm
distance at most $\epsilon$ from a separable state, which is necessarily far from the original EPR pair.

\end{proof}

\section{Non unital channels}

\begin{thm}
\label{thm:nonunital}
For any single-qubit density matrix $P \neq I/2$, there exists a threshold value $p_T$ such that if $C$ is any non-unital qubit channel with $\|C - I\|_\Diamond < p_T$ and fixed point $P$, then it is possible to perform any quantum computation of width $n$ and depth $D$ using $O(n \polylog(nD))$ physical qubits and depth $O(D \polylog(nD))$, where each physical qubit undergoes the noise $C$ at each time step.
\end{thm}

Note that the constants in the depth and number of physical qubits will depend on the channel $C$, and that the threshold may depend on $P$.

\begin{proof}
The basic idea is to perform a concatenated fault tolerance protocol, as usual for threshold proofs, but the fresh ancillas usually used for a threshold protocol will instead be provided by a separate part of the computer, set aside to act as a refrigerator.  In particular, the qubits in the computer will be broken up into 3 components:

\begin{enumerate}

\item \label{computation} The computation to be run, using error
correcting codes and assuming a fresh supply of qubits. This is possible using standard fault tolerant
techniques (see \cite{AB97,Kit97,KLZ98,AGP05}). This computation assumes the error rate per step is below the threshold for fault tolerance, so $p_T$ must be at least as small as the usual threshold.

\item \label{storage} A storage house for ancillas. This is just a
large (but linear in the size of the computation) array of chunks of
qubits. Qubits are going to stay in the storage house for a long time. This means
that they will undergo the noise $C$ many times before being used again.

\item \label{fridge} A refrigerator. This component takes a chunk
out of the storage house, and ``cools down'' a part of it, by encoding
all the information in the beginning of the chunk, and leaving a few
$\ket{0}$ qubits in the end. The cooling process does not
violate the second rule of thermodynamics, as the noisy channel leaves each qubit with non-maximal entropy. Rather the refrigerator
condenses the entropy using algorithmic cooling~\cite{SV99}, and supplies the computation (component 1)
with fresh bits.

\end{enumerate}

At the start of the computation, the ancilla qubits are taken from a fresh supply.  Once they have been used, ancilla qubits are then passed to the storage house where they will sit for a long time $T$.  After a time $T$ passes, groups of $R$ qubits are removed from the storage house and placed into the refrigerator.  The output of the refrigerator is 1 ``reset'' qubit and R-1 ``waste'' qubits.  The reset qubits are then passed to the computation section to act once again as ancillas and the waste qubits are returned to the storage to wait again.  For simplicity, let us assume the storage house begins full of qubits in the state $P$, so that the refrigerator can begin operating immediately, and that the initial supply of fresh ancilla qubits for the computation component only need last until the refrigerator finishes its first cycle.

\begin{figure}
\includegraphics*[scale=0.6]{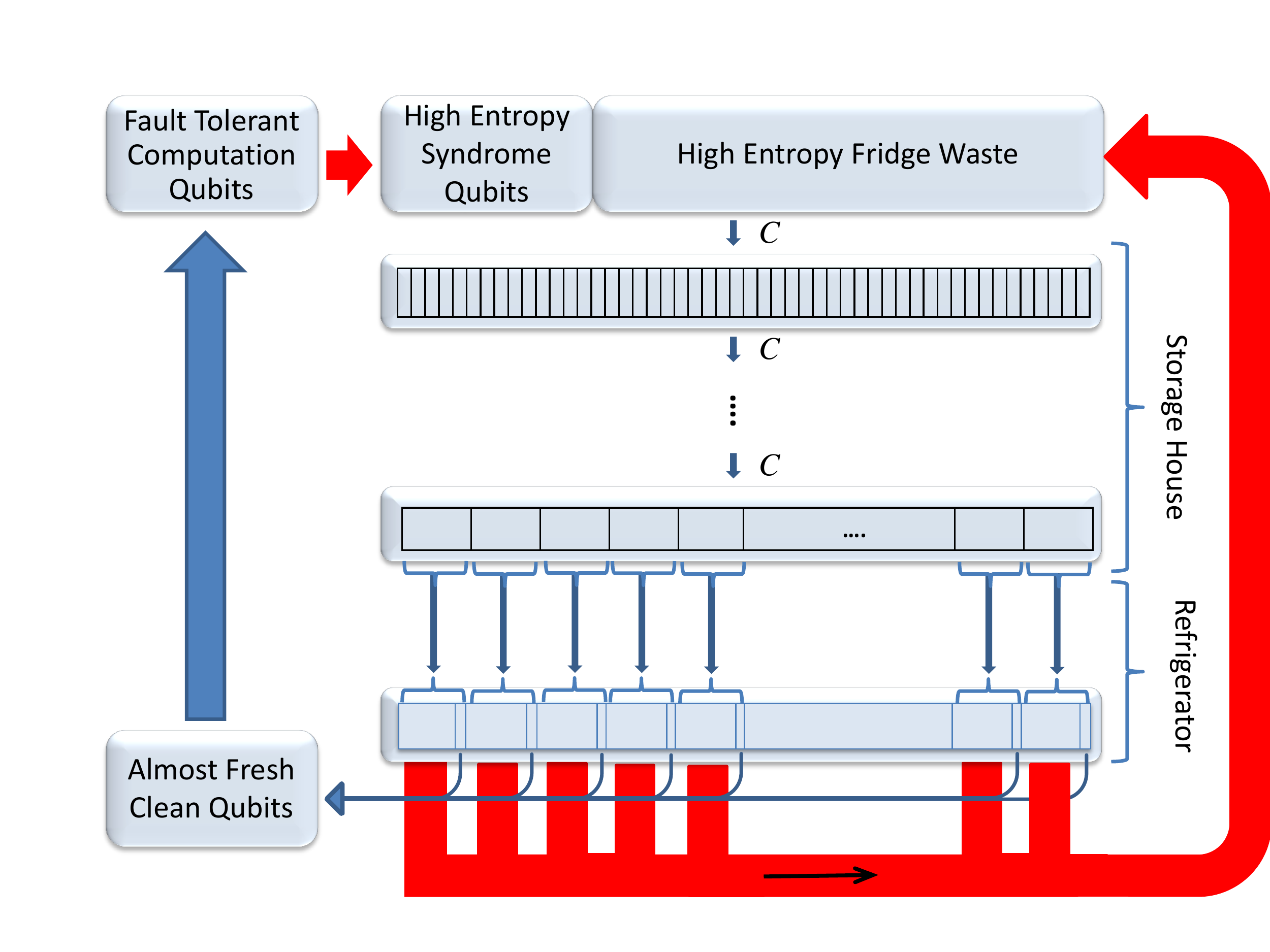}
\caption{The computer consists of three components: A computation component, a storage house, and a refrigerator.}
\end{figure}

The computation component runs a standard concatenated fault-tolerant protocol.  By the threshold theorem, if fresh ancillas are available, a protocol exists using $n' = O(n \polylog(nD/\epsilon_0))$ physical qubits and depth $D' = O(D \polylog(nD/\epsilon_0)$ which achieves a final state with $1$-norm distance at most $\epsilon_0$ from the correct final state of the computation.  Here ``fresh'' ancilla qubits does not mean they are perfect; it is sufficient that they are prepared by a process that is independent from ancilla to ancilla such that each prepared state has $1$-norm distance from $\ket{0} \bra{0}$ which is less than the threshold value. In our case, we do not have fresh ancillas; the ancillas will instead be provided by the other components of the protocol, and may have errors which are very slightly correlated with each other.  However, the computation component is designed just as if the ancillas were actually fresh.

The time $T$ is chosen so that any qubit placed in the storage house for a time $T$ (during which time it experiences no gates other than the channel noise $C$) reaches a state very close to the fixed point $P$ of $C$.  Let $C_P$ be the channel which throws away a qubit and replaces it with $P$: $C_P(\rho) = P$.  Then we choose $T$ so that $\|C^T - C_P \|_\Diamond < \epsilon_1/(n'D'R)$.  Since $C$ is strictly contractive, $C^T$ approaches $C_P$ exponentially, so we may choose $T = O(\log (n'D'R)) = O(\log (nDR))$ to achieve a constant value of $\epsilon_1$.

The refrigerator is a compression algorithm, chosen so that when run with an ideal circuit on $R$ qubits in the state $P$, the output is $1$ ``reset'' qubit in a state $R$ and $R-1$ ``waste'' qubits in any other state.  $R$ is chosen so that $\|R - \ket{0}\bra{0}\|_1 < \epsilon_2$.  The state of the waste qubits is not important, but typically will be close to maximally mixed, indicating we have optimally purified a state.  Also note that the waste qubits may be entangled with the reset qubit, but they cannot be very entangled, since $R$ is close to a pure state.  A refrigerator with these properties can be achieved with a constant value of $R$ and a quantum circuit of constant size $F$.  The size $F$ is a count of the total number of locations in the refrigerator circuit, including wait steps.  ($R$ and $F$ will in general depend on $P$, or at least on $S(P)$.)

Now, in the actual computation, the refrigerator experiences noise as well.  However, at each step, the noise is close to the identity.  A single run of the refrigerator algorithm with the real system therefore will therefore have distance at most $Fp_T$ from the ideal refrigerator.  Therefore, the output of one run of the refrigerator algorithm, if it is run with noise, but with all input qubits in exactly the state $P$, has distance at most $\epsilon_2 + Fp_T$ from the state $\ket{0} \bra{0}$.  Let us choose $\epsilon_2$ and $p_T$ so that $\epsilon_2 + Fp_T$ is below the threshold for fault tolerance.  Also note that at the start of the computation, ancilla qubits used before the first run of the refrigerator ends have experienced at most $F$ locations as well, so their error rates are also below the threshold.

Therefore, if the storage house is replaced by a black box which performs $C_P$ on each qubit placed in it for a time $T$, the refrigerator outputs qubits which are sufficiently close to perfect to satisfy the requirements of the threshold theorem.  Also, since the black box treats each qubit in the storage house separately, and the refrigerator algorithm is independent for each reset qubit, the outputs of the refrigerator will, in this version of the protocol, be independent.  Therefore, when there is a black box storage house, the protocol satisfies the conditions of the threshold theorem, and the final state of the protocol has distance at most $\epsilon_0$ from the correct output of an ideal protocol without noise.

However, the actual effect of the storage house is $C^T$ on each qubit, which is very close to $C_P$.  Therefore, the true protocol will be very close to the one with a black box storage house.  How close?  To figure this out, we need to determine the total number $M$ of qubits that get put through the storage house.  Since $\|C^T - C_P \|_\Diamond < \epsilon_1/(n'D'R)$ for a single qubit, the full protocol will be a distance at most $M \epsilon_1/(n'D'R)$ from the one with a black box storage house.  When $M$ is not too large, even though there may be residual correlations between the ancilla qubits exiting the storage house, the correlations are too weak to have a big effect on the protocol.

At each time step, the computation component uses at most $n'$ ancillas, each of which must be output by a refrigerator protocol which uses $R$ qubits from the storage house.  Therefore, we need at most $n'R$ qubits to come out of the storage house at each time step, and there are at most $D'$ time steps in the full fault-tolerant protocol, so $M \leq n'RD'$.  The overall protocol thus has distance at most $\epsilon_1$ from the black box protocol, which in turn outputs a final state with distance at most $\epsilon_0$ from the ideal protocol.

By choosing $\epsilon_0 = \epsilon_1 = \epsilon/2$ for any constant $\epsilon$, we therefore get a fault-tolerant protocol which outputs a state within distance $\epsilon$ of the correct output state, as desired.  At any time, the computation component uses $n'$ qubits, the storage house contains at most $TRn'$ qubits, and the refrigerator contains at most $FRn'$ qubits.  Therefore the total number of qubits used is at most $(1+TR + FR)n' = O(Tn') = O(n \polylog(nD))$.  The depth blow-up of the computation is the same as for the usual fault-tolerant protocol; that is, the depth is $D' = O(D \polylog(nD))$.
\end{proof}

While this protocol works to simulate a circuit of any depth $D$, there is an overhead cost that depends polylogarithmically on $D$.  Therefore, if $D$ is superexponential, the number of qubits needed will be superpolynomial.  Probably it is impossible to compute for more than exponential time; for instance, for the amplitude damping channel, there is a probability at least $p^n$ at each time step that all qubits will relax to $\ket{0}$, leaving no remaining information about the input to the computation.  However, we have not yet been able to prove this for a general non-unital channel.

%
%

It is worth noting that theorem~\ref{thm:nonunital} still holds when quantum gates are restricted to nearest-neighbor interactions with a two-dimensional geometry, but we are not certain if it holds or not for a one-dimensional architecture.  To implement the procedure in two dimensions, we lay out the ``computation'' component of the system along a line.  It is possible to do this and still have a threshold error rate for fault tolerance~\cite{AB08,Got00}.  The ancilla qubits needed for operations on each block of the quantum error-correcting code used in the protocol are prepared in a strip of constant width laid out orthogonal to the computation qubits, as in figure~\ref{fig:local}.  Each strip contains a separate ``storage'' component and ``refrigerator'' component.  In the context of a concatenated code, only the ancillas needed for operations at the lowest-level blocks need to be reset --- higher-level ancilla qubits are built up from individual physical qubits, and when the physical qubits are nearly pure, the higher-level ancillas can be made from them in the usual way.

\begin{figure}
\includegraphics*[scale=0.6]{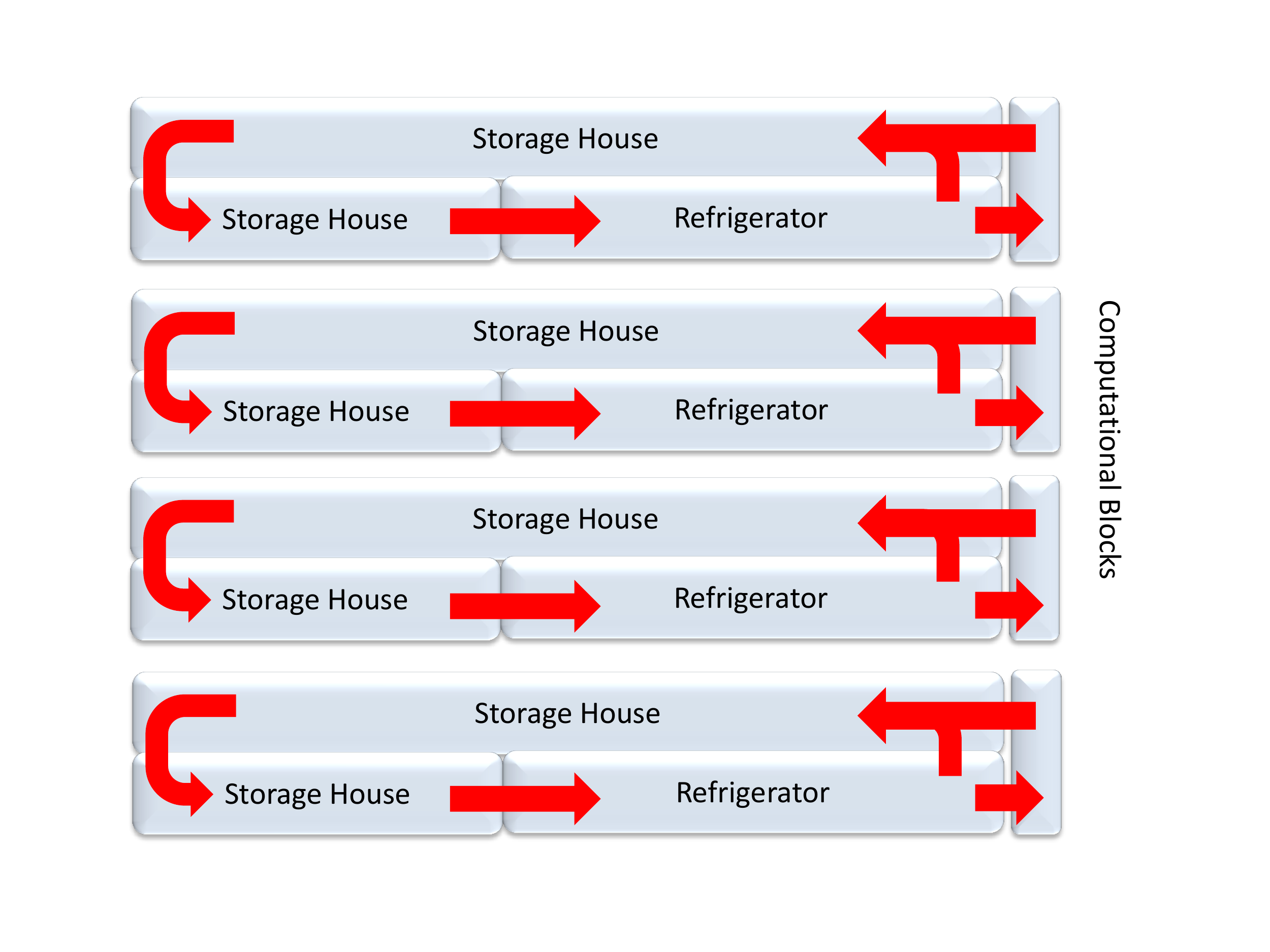}
\caption{Each block of the code has a separate storage house and refrigerator for a two-dimensional layout.}
\label{fig:local}
\end{figure}

Each strip is itself divided into two halves.  In one half, qubits are moved away from the computation component by one place each time step.  In the other half, qubits are moved one step closer to the computational component at each time step.  At the end of the strip away from the computation component, qubits are shifted from the half of the strip moving away to the half moving towards the computational component.  Thus, if there were no other operations going on, qubits would travel to the end of the strip and then come back.  The part of the strip moving away, and perhaps part of the half returning to the computational component, is devoted to the storage component.  Qubits in it undergo no operations beyond moving along the strip.  We make the strip sufficiently long so that qubits remain in the storage house for a long enough time $T$ to be close to the fixed point $P$.

The last part of the returning half of the strip implements the refrigerator component.  The algorithmic cooling used in the refrigerator uses only a constant size circuit to prepare a constant number of ancillas (which is all that is needed for a single time step for a single block of the quantum error-correcting code), so a strip of constant width is adequate for this purpose.  Once qubits in the strip have reached the computational component, they have been separated out into ``reset'' qubits, which are almost pure and can be used as ancillas in the computation, and ``waste'' qubits, which are returned to the storage house in the outgoing half of the strip.

In one dimension, this strategy doesn't appear to work.  We could attempt to intersperse level $1$ computational blocks with storage and fridge blocks, but as the error rate gets small, the storage house must get bigger and bigger since it takes longer for qubits to relax to the fixed point of the channel.  This means the computational blocks get further and further apart, making it costly to do gates between them, which in turn means the error rate for logical gates is higher, offsetting the advantage of the fault-tolerant protocol.  It is not clear if there is a cleverer approach which allows exponentially long computations in one dimension with non-unital noise.

Note also that in this design, we are assuming that the original circuit to be simulated is laid out using nearest-neighbor gates in one dimension, yet we are implementing it in two dimensions.  If the ideal circuit is instead given to us in a two- or higher-dimensional layout, we have two choices.  Either we can use the approach above and add one dimension to the layout for the storage house and refrigerator components, or we can rewrite the ideal circuit into a one-dimensional one and then implement it in two dimensions.  The latter choice potentially incurs an $O(n)$ blow-up in the depth of the circuit.  It may be that there is an alternative approach to simulate two-dimensional circuits without fresh ancillas without adding a dimension and with only polylog overhead.

\section{Conclusion}

We have shown that the usual claim that a quantum computation needs fresh ancilla qubits supplied from the outside is not completely correct.  Rather, we have shown that if the noise in the system is non-unital, then we can co-opt the noise to create fresh ancilla qubits, allowing us to recycle qubits and perform computation even in an otherwise closed system.  We can do this because non-unital noise provides some cooling, taking a maximally mixed state to a less mixed state.  For instance, if the system is in thermal contact with a cold bath, qubits left in the storage house will tend to equilibrate to the temperature of the bath.  Viewed this way, of course the system is not truly closed, but it is open in only a single limited way.

Moreover, our protocol works for \emph{any} fixed point of the non-unital map.  Therefore, even a very hot bath, provided it is at finite temperature, can provide enough cooling to leverage into a full-blown fault-tolerant computation.  Of course, the overhead is very high when the fixed point $P$ is close to maximally mixed, but it is perhaps surprising that it can still be done at all.  One might instead have expected a threshold temperature above which the system is too hot to compute.  There is a threshold \emph{error strength}, but that is a constraint on the coupling to the bath rather than the temperature of the bath.

The dephasing class of channels is another interesting case, sitting on the borderline between the depolarizing class, for which long computations are impossible, and the amplitude damping class, for which computation is possible.  A pure dephasing channel, with no other sources of noise at all, is physically improbable, but in many physical systems, dephasing is the dominant error source.  This suggests that in early quantum computers, it might be sufficient to keep a stock of ancilla qubits from the start and use them up over the course of the computation --- the polynomial overhead we derive for this case --- and only resort to resetting ancilla qubits either via external control or internal refrigeration once larger quantum computers are built.  Indeed, this is the strategy already envisioned for many early quantum computers.  The dephasing class is also subject to an interesting gap between the upper and lower bounds.  Certainly it is possible to store a single EPR pair for $O(n/\log n)$ time using a good QECC, but we have only been able to show that the upper bound on the storage time is $O(n^3)$.  Can this upper bound be improved, or is there some storage procedure that is better than the obvious one?

Finally, we have not considered channels for qudits of dimension greater than $2$.  We expect essentially the same classification to apply, based essentially on the ``best'' two-dimensional subspace for the channel.  Channels which have a single fixed point at the maximally mixed state should still make computation impossible, whereas non-unital channels should still provide some cooling and thus make it possible to run exponentially long quantum computations.  Unital channels which leave some pure states unchanged should act like the dephasing class, allowing polynomially-long quantum computations.  However, qudit channels are more complicated than qubit channels, so we have not proven that these are all the possible behaviors.

\paragraph{Acknowledgments}
We would like to thank an anonymous referee for QIP, who asked if the protocol still worked with local gates.  Michael Ben-Or is the incumbent of the Jean and Helena Alfassa Chair in Computer Science and a member of the I-CORE Center of Excellence in Algorithms.  His research was supported in part by the Israeli Science Foundation (ISF) research grant 1446/09.  Research at Perimeter Institute is supported by the Government of Canada through Industry Canada and by the Province of Ontario through the Ministry of Economic Development and Innovation.  Daniel Gottesman is also supported by the CIFAR Quantum Information Processing program.  Avinatan Hassidim is supported by GIF Young grant 2308/2011 and by the ISF, grant number 1241/12.



\begin{thebibliography}{99}



\bibitem[AB97]{AB97} D. Aharonov and M. Ben-Or, ``Fault-Tolerant Quantum Computation
With Constant Error,'' In {\it
  Proceedings of the 29th Annual IEEE Symposyum on the Theory of
Computer Science} (STOC'97),
   ACM,   pp. 176-188, 1997.

\bibitem[AB08]{AB08} D. Aharonov and M. Ben-Or, ``Fault-Tolerant Quantum Computation with Constant Error Rate,''
{\it SIAM J. Comput.}, Vol. 38, pp. 1207-–1282, 2008, quant-ph/9906129.

\bibitem[ABIN96]{ABIN96} D. Aharonov, M. Ben-Or, R. Impagliazzo and N. Nisan
`` Limitations of Noisy Reversible Computation,''
quant-ph/9611028, 1996.

\bibitem[AGP05]{AGP05} P. Aliferis, D. Gottesman, and J. Preskill, ``Quantum accuracy threshold for concatenated distance-3 codes,'' {\it Quant. Information and Computation} Vol. 6, No. 2, pp. 97--165, 2006, quant-ph/0504218.

%
%
%
%
%

\bibitem[Got00]{Got00} D. Gottesman, ``Fault-Tolerant Quantum Computation with Local Gates,'' {\it J. Modern Optics} Vol. 47, pp. 333--345, 2000, quant-ph/9903099.

\bibitem[KR01]{KR01} C. King and M. Beth Ruskai, ``Minimal entropy of states emerging from noisy quantum channels,''
{\it
  IEEE Transactions on Information Theory} Vol. 47, No. 1, pp. 192--209,2001.

\bibitem[Kit97]{Kit97} A. Yu. Kitaev, ``Quantum computations: algorithms and error correction,'' {\it Russian Math. Surveys} Vol. 52, pp. 1191--1249,
1997.

\bibitem[KSV02]{KSV02} A.Y. Kitaev and A.H. Shen and M.N. Vyalyi, {\it Classical and Quantum Computation}, AMS (Providence, RI), 2002.

\bibitem[KLZ98]{KLZ98} E. Knill, R. Laflamme, W. H. Zurek, ``Resilient quantum computation: error models and thresholds,'' {\it Proc. Roy.
Soc. London, Ser. A} Vol.454, p. 365, quant-ph/9702058, 1998.

%
%
%
%
%

\bibitem[Raz06]{Raz06} A. A. Razborov, ``An Upper Bound on the Threshold Quantum
Decoherence Rate,''  quant-ph/0310136,2006.



\bibitem[SV99]{SV99} L. Schulman and U. Vazirani, ``Molecular scale heat engines and scalable quantum computation,''
{\it Proc. STOC '99}, pp. 322--329, 1999.

\bibitem[TB05]{TB05} B. M. Terhal and G. Burkard, ``Fault-Tolerant Quantum Computation For Local Non-Markovian Noise,''
{\it Phys. Rev. A} Vol. 71, 012336, 2005, arXiv:quant-ph/0402104.

\bibitem[Wat08]{Watrous} J. Watrous, lecture notes for ``Theory of Quantum Information'' (CS 798),
ch.~8, https://cs.uwaterloo.ca/\~{}watrous/quant-info/lecture-notes/08.pdf, 2008.

\end{thebibliography}
\end{document}